\documentclass[a4paper, 10pt, twosides]{amsart}

\usepackage{lmodern}
\usepackage{amssymb}
\usepackage{amsthm}
\usepackage{mathtools}
\usepackage{MnSymbol}

\usepackage{enumitem}
\setitemize{nolistsep}
\setenumerate{nolistsep}
\usepackage{todonotes}
\allowdisplaybreaks

\newtheorem{fact}{Fact}
\newtheorem{thm}{Theorem}
\newtheorem{defn}{Definition}
\newtheorem{assump}{Assumption}
\newtheorem{rem}{Remark}

\newtheorem{mech}{Mechanism}
\newtheorem{ex}{Example}

\newtheorem{lem}{Lemma}

\newcommand{\norm}[1]{\left\lVert{#1}\right\rVert}
\newcommand{\abs}[1]{\left\lvert{#1}\right\rvert}
\newcommand{\lra}{\longrightarrow}
\newcommand{\Let}{\coloneqq}

\newcommand{\pmat}[1]{\begin{pmatrix}#1\end{pmatrix}}

\renewcommand{\geq}{\geqslant}

\renewcommand{\leq}{\leqslant}

\newcommand{\R}{\mathbb{R}}
\newcommand{\N}{\mathbb{N}}
\renewcommand{\P}{\mathcal{P}}

\newcommand{\K}{\mathcal{K}}
\newcommand{\KL}{\mathcal{KL}}
\newcommand{\Kinfty}{\mathcal{K}_{\infty}}
\newcommand{\Ntsigma}{\mathrm{N}_{\sigma}(0,t)} 

\title[]{A graph theoretic approach to input-to-state stability of switched systems}

\author[A.\ Kundu]{Atreyee Kundu}
\address{Control Systems Technology Group, Department of Mechanical Engineering, Eindhoven University of Technology, 5600MB Eindhoven, The Netherlands. }
\email[A.\ Kundu]{a.kundu@tue.nl}
\author[D.\ Chatterjee]{Debasish Chatterjee}
\address{Systems \& Control Engineering, Indian Institute of Technology Bombay, Mumbai~--~400076, India}
\email[D.\ Chatterjee]{dchatter@iitb.ac.in}
\urladdr[D.\ Chatterjee]{{http://www.sc.iitb.ac.in/~chatterjee}}

\date{\today}

\begin{document}



\begin{abstract}
    This article deals with input-to-state stability (ISS) of discrete-time switched systems. Given a family of nonlinear systems with exogenous inputs, we present a class of switching signals under which the resulting switched system is ISS. We allow non-ISS systems in the family and our analysis involves graph-theoretic arguments. A weighted digraph is associated to the switched system, and a switching signal is expressed as an infinite walk on this digraph, both in a natural way. Our class of stabilizing switching signals (infinite walks) is periodic in nature and affords simple algorithmic construction.
\end{abstract}

%
%
  \maketitle
\section{Introduction}
\label{s:intro}
    A \emph{switched system} comprises of two components --- a family of systems and a switching signal. The \emph{switching signal} selects an \emph{active} subsystem at every instant of time, i.e., the system from the family that is currently being followed \cite[\S1.1.2]{Liberzon}. In this article we study ISS of discrete-time switched systems under \emph{constrained switching} \cite[Chapter 3]{Liberzon}. More specifically, given a family of discrete-time systems with exogenous inputs such that not all systems in the family are ISS, we are interested in identifying a class of switching signals under which the resulting switched system is ISS.

    For a given family of discrete-time systems, in \cite{knc_hscc14} we proposed a class of switching signals under which the resulting switched system is globally asymptotically stable (GAS). We admitted unstable subsystems and our stabilizing condition involved \emph{only} certain asymptotic properties of the switching signals. Although the said result was presented in the context of switched linear systems for simplicity, it extends readily to the nonlinear setting under standard assumptions.

    Algorithmic synthesis of the class of stabilizing switching signals presented in \cite{knc_hscc14} was studied in \cite{knc_hscc14, kbnc_arxiv14}. A weighted digraph was associated to the given family of systems and the admissible transitions, and the switching signal was expressed as an infinite walk on the above digraph. In this setting, given a family of systems, algorithmic construction of a stabilizing switching signal is identical to: given the underlying weighted digraph of a switched system, algorithmic construction of an infinite walk that satisfies a certain pre-specified condition. However, algorithmically constructing an infinite walk on a given weighted digraph, that satisfies a pre-specified condition involving vertex and edge weights is an infeasible problem (because an algorithm should terminate in finite time). As a natural alternative in \cite{knc_hscc14, kbnc_arxiv14} we chose to construct the desired infinite walk by repeating a suitable closed walk. More specifically, we used periodic construction of infinite walks that correspond to stabilizing switching signals.

    Given a family of discrete-time nonlinear systems with exogenous inputs, in this article we extend the above periodic construction of infinite walks to ISS of the resulting switched system. In particular, the main features of our result are twofold:
    \begin{itemize}[label = $\circ$, leftmargin = *]
    	\item Firstly, we allow non-ISS systems in the family, and
	\item Secondly, our class of switching signals affords a simple algorithmic construction.
    \end{itemize}
        
        In general, given a family of systems, algorithmically constructing a switching signal that obeys point-wise constraints on the number of switches and the duration of activation of subsystems \cite[Chapter 3]{Liberzon}, \cite{HespanhaMorse, chatterjee07, BaiISS11, Liberzon_IOSS, kcnl_arxiv15} is not an easy task because the stabilizing conditions need to be checked for \emph{every} interval of time. In \cite{Mitra_ADT} the authors proposed methods for verifying (checking) average dwell time by expressing the switching signal as an infinite execution of a hybrid automaton. In this article we opt for switching signals that are of periodic nature; this periodicity ensures simpler algorithmic construction as compared to switching signals with point-wise constraints.

        Observe that a periodic switching signal in the discrete-time setting necessarily obeys an average dwell time condition. But unlike average dwell time switching, we do not impose separate point-wise constraints on the number of switches and the duration of activation of non-ISS subsystems \cite{Liberzon_IOSS}. Our stability condition solely relies on periodic validity of \emph{an} inequality involving certain parameters of the subsystems and the switching signal.
        
        We employ graph-theoretic arguments as the main apparatus for our analysis.\footnote{Digraphs have appeared before in the switched systems literature in \cite{Mancilladigraph, graph_dwell, Lazar14, Lee06, Langerak03, knc_hscc14, kbnc_arxiv14}.} Given the underlying weighted digraph of a switched system, our class of stabilizing switching signals correspond to infinite walks that admit periodic construction in terms of suitable closed walks. Consequently, the algorithmic construction of a stabilizing switching signal consists of two steps --- first, constructing a closed walk satisfying a pre-specified condition, and second, a mechanism to repeat the above closed walk indefinitely many times. We discuss standard graph-theoretic algorithms from the literature to execute the above algorithmic construction.

    The remainder of this article is organized as follows: In \S\ref{s:probstate} we formulate the problem under consideration, and catalog certain preliminaries which would be used in our analysis. Our main result appears in \S\ref{s:mainres}. We also discuss various features of our main result through a series of remarks in this section. We provide a numerical example in \S\ref{s:numex} and conclude in \S\ref{s:concln}. The proof of our main result appears in \S\ref{s:proofs}.
    
    {\bf Notation}: $\N$ is the set of natural numbers \(\{1, 2, \ldots\}\), $\N_{0} = \N\cup\{0\}$. We denote by $\norm{v}$ the standard Euclidean norm of a vector $v$, while $\norm{w}_{t} \Let \sup\{\norm{w(t)}:t\in\N_0\}$ denotes the supremum norm of a signal \(w\) taking values in some Euclidean space. For a walk $W$ on a digraph $G(V,E)$, $\abs{W}$ denotes the length of $W$.

\section{Problem Statement}
\label{s:probstate}
\subsection{The switched system}
\label{ss:swsys}
    We consider a family of discrete-time systems with exogenous inputs
    \begin{align}
    \label{e:family}
        x(t+1) = f_{i}(x(t),v(t)),\:\:x(0)\:\text{given},\:\:i\in\P,\:\:t\in\N_{0},
    \end{align}
    where $x(t)\in\R^{d}$ is the vector of states, and $v(t)\in\R^{m}$ is the vector of inputs at time $t$, $\P = \{1,2,\cdots,N\}$ is an index set. We assume that for each $i\in\P$, $\ker f_{i}(\cdot,0) = \{0\}$. Let $\sigma:\N_{0}\rightarrow\P$ be a switching signal that specifies, at every time $t$, the index of the active system from the family \eqref{e:family}. The \emph{discrete-time switched system} generated by the given family of systems \eqref{e:family} and the switching signal $\sigma$ is given by
    \begin{align}
    \label{e:swsys}
        x(t+1) = f_{\sigma(t)}(x(t),v(t)),\:\:x(0)\:\text{given},\:\:t\in\N_{0}.
    \end{align}
    Let $0 =: \tau_{0}<\tau_{1}<\cdots$ be the \emph{switching instants} of $\sigma$; these are the integers at which $\sigma$ jumps. We let $(x(t))_{t\in\N_{0}}$ denote the solution to the switched system \eqref{e:swsys}, where the dependence on $\sigma$ is suppressed for notational simplicity. We assume that there are no jumps in the states at the switching instants.


    \begin{defn}[{\cite[Definition 3.1]{Sontag_dtiss}}]
    \label{d:iss}
        The switched system \eqref{e:swsys} is input-to-state stable (ISS) for a given $\sigma$ if there exist functions $\beta\in\KL$ and $\gamma\in\K$ such that for all bounded inputs $v:\N_0\lra\R^m$ and $x(0)\in\R^{d}$, we have\footnote{We refer the reader to \cite[\S 4.4]{Khalil} for definitions of classes $\K$, $\Kinfty$, $\mathcal{L}$ and $\KL$ functions.}
        \begin{align}
        \label{e:iss}
            \norm{x(t)} \leq \beta(\norm{x(0)},t) + \gamma(\norm{v}_{t})\:\:\text{for all}\:\:t\in\N_{0}.
        \end{align}
        If no inputs are present, i.e., $v\equiv 0$, then \eqref{e:iss} reduces to global asymptotic stability (GAS) of \eqref{e:swsys}.
    \end{defn}

    Let $\P_{S}$ and $\P_{U}$ denote the sets of indices of ISS and non-ISS systems in family \eqref{e:family}, respectively, $\P = \P_{S}\sqcup\P_{U}$. Let the set $E(\P)$ consist of all pairs $(i,j)$ such that it is allowed to switch from system $i$ to system $j$, $i,j\in\P$. 
    \begin{assump}
    \label{a:lambdaj}
        For each $i\in\P$, there exist continuous functions $V_{i}:\R^{d}\lra[0,+\infty[$, class $\Kinfty$ functions $\underline{\alpha}$, $\overline{\alpha}$, class $\K$ function $\gamma$ and scalars $\lambda_{i}$ with $0<\lambda_{i}< 1$ for $i\in\P_{S}$ and $\lambda_{i} > 1$ for $i\in\P_{U}$ such that for all $\xi\in\R^{d}$ and $\eta\in\R^{m}$, we have
        \begin{align}
        \label{e:Lyapprop}
            \underline\alpha(\norm{\xi}) \leq V_{i}(\xi) \leq \overline\alpha(\norm{\xi}),
        \end{align}
        and
        \begin{align}
        \label{e:dLyapprop}
            V_{i}(f_{i}(\xi,\eta)) \leq \lambda_{i}V_{i}(\xi) + \gamma(\norm{\eta}),\:\:t\in\N_{0}.
        \end{align}
    \end{assump}
        The functions $(V_{i})_{i\in\P}$ satisfying conditions \eqref{e:Lyapprop} and \eqref{e:dLyapprop} are called the ISS-Lyapunov-like functions and are standard in the literature, see e.g., \cite{Sontag_dtiss}, \cite{Grune2014} for details regarding existence of such functions and their properties.
    \begin{assump}
    \label{a:muij}
         Whenever $(i,j)\in E(\P)$, there exist $\mu_{ij} > 0$ such that the ISS-Lyapunov-like functions are related as follows:
        \begin{align}
        \label{e:muijineq}
            V_{j}(\xi) \leq \mu_{ij}V_{i}(\xi)\:\:\text{for all}\:\:\xi\in\R^{d}.
        \end{align}
    \end{assump}
        The assumption of linearly comparable Lyapunov-like functions, i.e., there exists $\mu\geq 1$ such that
        \begin{align}
        \label{e:muineq}
            V_{j}(\xi) \leq \mu V_{i}(\xi)\:\:\text{for all}\:\:\xi\in\R^{d}\:\:\text{and all}\:\:i,j\in\P,
        \end{align}
        is standard in the theory of stability under average dwell time switching, see e.g., \cite{Zhai002}. Clearly, \eqref{e:muineq} is a special case of \eqref{e:muijineq}.

\subsection{The underlying weighted digraph}
\label{ss:wdigraph}
    We associate a weighted digraph $G(\P,$\\$E(\P))$ with the switched system \eqref{e:swsys} in the following manner:
    \begin{itemize}[label = $\circ$, leftmargin = *]
        \item The set of vertices is the set of indices $\P$.
        \item The set of edges $E(\P)$ consists of:
        \begin{itemize}[label = $\diamond$, leftmargin = *]
            \item a directed edge $(i,j)$ whenever it is allowed to switch from vertex (system) $i$ to vertex (system) $j$, $i,j\in\P$,
            \item a self-loop at vertex $j$ whenever it is allowed to dwell on vertex (system) $j$ for two or more consecutive time-steps.
        \end{itemize}
        \item The parameters $\abs{\ln\lambda_{j}}$'s, $j\in\P$ ($\grave{a}$ la Assumption \ref{a:lambdaj}) and $\ln\mu_{ij}$, $(i,j)\in E(\P)$ ($\grave{a}$ la Assumption \ref{a:muij}) are the vertex and edge weights of $G(\P,E(\P))$, respectively. It is evident that $\ln\mu_{jj} = 0$.
    \end{itemize}

    Recall that \cite[p.\ 4]{Bollobas} a \emph{walk} on a digraph $G(V,E)$ is an alternating sequence of vertices and edges, say $v_{0}, e_{1}, v_{1}, e_{2}, \cdots, e_{\ell}, v_{\ell}$, where $v_{i}\in V$, $e_{i} = (v_{i-1},v_{i})\in E$, $0 < i \leq \ell$. A walk is \emph{closed} if $v_{0} = v_{\ell}$. The length of a walk is its number of edges, counting repetitions, e.g., in the above case the length of the walk $W$ is $\ell$. In the sequel by the term \emph{infinite walk} we mean a walk of infinite length, i.e., it has infinitely many edges. We have the following:
    \begin{fact}[{\cite[Fact 3]{knc_hscc14}}]
    \label{fact:walk}
        The set of switching signals $\sigma: \N_{0}\lra\P$ and the set of infinite walks on $G(\P,E(\P))$ (defined as above) are in bijective correspondence.
    \end{fact}
        Observe that since we are in the discrete-time setting, the association of time with the length of a walk is natural.
    \begin{ex}
    \label{ex:wgraph}
    \rm
        \begin{itemize}[label = $\circ$, leftmargin = *]
            \item Consider a family of systems $\P = \{1,2,3\}$. Let the following switches be admissible: $1\to 2$, $2\to 3$, $3\to 1$, and $3\to 2$. Let it also be allowed to dwell on systems $1$ and $3$ for two (or more) consecutive time steps. A possible choice of switching signal $\sigma$ is: $\sigma(0) = 1$, $\sigma(1) = 1$, $\sigma(2) = 2$, $\sigma(3) = 3$, $\sigma(4) = 2$, $\sigma(5) = 3$, $\sigma(6) = 3$, $\cdots$
            \item The underlying weighted digraph $G(\P,E(\P))$ of the above switched system is: $\P = \{1,2,3\}$, and $E(\P) = \{(1,1),(1,2),(2,3),(3,1),(3,2),(3,3)\}$. The quantities $\abs{\ln\lambda_{1}}$, $\abs{\ln\lambda_{2}}$, $\abs{\ln\lambda_{3}}$ and $\ln{\mu_{11}}$, $\ln{\mu_{12}}$, $\ln{\mu_{23}}$, $\ln{\mu_{31}}$, $\ln{\mu_{32}}$, $\ln{\mu_{33}}$ (where $\lambda_{i}$'s, $i\in\P$ and $\mu_{ij}$, $(i,j)\in E(\P)$ are as in \eqref{e:dLyapprop} and \eqref{e:muijineq}, respectively) are associated as weights corresponding to vertices and edges, respectively. The infinite walk corresponding to the said switching signal is: $1,(1,1),1,(1,2),2,(2,3),3,(3,2),2,$\\$(2,3),3,(3,3),3,\cdots$
        \end{itemize}
    \end{ex}
    
    For a walk $W = v_{0},(v_{0},v_{1}),v_{1},(v_{1},v_{2}),v_{2},\cdots$ on $G(\P,E(\P))$, we define the quantity
    \begin{align}
    \label{e:Xidefn}
    	\Xi(W) := \sum_{(k,\ell)\in E(\P)} \bigl(\ln\mu_{k\ell}-1_{\{k\in\P_{S}\}}\abs{\ln\lambda_{k}}
	+1_{\{k\in\P_{U}\}}\abs{\ln\lambda_{k}}\bigr)\#\{k\rightarrow\ell\}_{W},
    \end{align} 
    where $\#\{k\rightarrow\ell\}_{W}$ denotes the number of times an edge $(k,\ell)\in E(\P)$ appears in $W$, $\ln\mu_{k\ell}$ and $\abs{\ln\lambda_{k}}$ are weights associated to an edge $(k,\ell)\in E(\P)$ and a vertex $k\in\P$, respectively.
    
    \begin{ex}
    \label{ex:XiforW}
    \rm
    	Consider the switched system and its underlying weighted digraph from $G(\P,E(\P))$ from Example \ref{ex:wgraph}. Let $\P_{S} = \{1,2\}$ and $\P_{U} = \{3\}$. Consider the closed walk $W = 3,(3,2),2,(2,3),3$. Consequently, $\displaystyle{\Xi(W) = (\ln\mu_{32}+\abs{\ln\lambda_{3}})+(\ln\mu_{23}}$\\$\displaystyle{-\abs{\ln\lambda_{2}})}$.
    \end{ex}
    
    \begin{defn}
    \label{d:contra}
    	A walk $W$ on $G(\P,E(\P))$ is called contractive if it satisfies
	\begin{align}
	\label{e:contrawalk}
		\Xi(W) < 0.
	\end{align}
    \end{defn}
    
    We next describe a mechanism to generate an infinite walk on $G(\P,E(\P))$ by repeating a (finite) closed walk. The requirement of generating an infinite walk in terms of a closed walk is at the level of algorithmic construction and its importance in our context will be clear in \S\ref{s:mainres}.
    \begin{mech}
    \label{mech:repeatmech}
    	Consider a finite closed walk $W' = v_{0},(v_{0},v_{1}),v_{1},\ldots,v_{n-1},(v_{n-1},v_{0}),$\\$v_{0}$ of length $n > 0$ on $G(\P,E(\P))$. We build an infinite walk $W$ by repeating $W'$ infinitely many times in the following manner: $v_{0},(v_{0},v_{1}),v_{1},\ldots,v_{n-1},(v_{n-1},v_{0}),$\\$v_{0}, (v_{0},v_{1}),v_{1},\ldots,v_{n-1},(v_{n-1},v_{0}),v_{0}, \ldots$
    \end{mech}
    
    \begin{rem}
    \label{r:contradef}
    \rm
    	Observe that for a walk $W$ on $G(\P,E(\P))$, the definition of $\Xi(W)$ excludes the weight of the final vertex of $W$ (i.e., the number of times a vertex is visited is considered to be the same as the total number of times its outgoing edges are visited). This is however no loss of generality since our focus is on infinite walks constructed by repeating ($\grave{a}$ la Mechanism \ref{mech:repeatmech}) a closed contractive walk $W$ on $G(\P,E(\P))$.
    \end{rem}

\section{Main Result}
\label{s:mainres}
    We are now in a position to present our main result, a detailed proof of which is presented in \S\ref{s:proofs}.
    \begin{thm}
    \label{t:mainthm}
        Consider the underlying weighted digraph $G(\P,E(\P))$ of a switched system \eqref{e:swsys} as described in \S\ref{ss:wdigraph}. The switched system \eqref{e:swsys} is input-to-state stable (ISS) for every switching signal $\sigma$ whose corresponding ($\grave{a}$ la Fact \ref{fact:walk}) infinite walk $W$ is obtained by repeating ($\grave{a}$ la Mechanism \ref{mech:repeatmech}) a closed contractive walk $W'$ on $G(\P,E(\P))$.
    \end{thm}

        Given a family of systems such that not all subsystems are ISS, the above theorem identifies a class of switching signals under which the resulting switched system is ISS. A switching signal $\sigma$ which is a member of the said class of stabilizing switching signals, corresponds ($\grave{a}$ la Fact \ref{fact:walk}) to an infinite walk $W$ that is constructed by repeating ($\grave{a}$ la Mechanism \ref{mech:repeatmech}) a closed contractive walk $W'$ on $G(\P,E(\P))$ --- the underlying weighted digraph of the switched system \eqref{e:swsys}. Consequently, a stabilizing switching signal is periodic in nature with the period being equal to the length of the closed contractive walk. See \S\ref{s:proofs} for a detailed proof of the above theorem.

   \begin{ex}
   \label{ex:themx}
   \rm
   	Consider the switched system and its underlying weighted digraph from Example \ref{ex:wgraph}, and a closed contractive walk $W' = 3, (3,2), 2, (2,3), 3$. We construct an infinite walk $W$ by repeating the closed contractive walk $W'$, i.e., $W = 3, (3,2), 2, (2,3), 3, (3,2), 2, (2,3), 3, (3,2), 2, \cdots$. According to Theorem \ref{t:mainthm}, the switched system under consideration is ISS under a switching signal $\sigma$ corresponding to the infinite walk $W$.
   \end{ex}
	In the remainder of this section we elaborate on various features of Theorem \ref{t:mainthm}.
    \begin{rem}
    \rm
        The contractivity condition in \eqref{e:contrawalk} for a closed walk $W'$ can be rewritten as $\Xi(W') \leq -\varepsilon$ for some $\varepsilon > 0$. Consequently, how ``contractive'' the walk $W'$ is, depends on how large $\varepsilon$ is. The contractivity of the closed walk $W'$ corresponds to the ``stability margin'' of the switching signal $\sigma$.
    \end{rem}
    \begin{rem}
    \label{r:adtcompa}
    \rm
        Prior results on ISS of switched systems involve point-wise constraints on the number of switches and the duration of activation of subsystems, see e.g., \cite{chatterjee07, Liberzon_IOSS, BaiISS11, kcnl_arxiv15}. On the one hand, given a family of systems, such conditions allow us to guarantee ISS provided that the switching signal obeys some pre-specified conditions on the rate of switching. For example, let Assumption \ref{a:lambdaj} hold with $\lambda_{j} = \lambda_{S}$ for all $j\in\P_{S}$, $\lambda_{k}=\lambda_{U}$ for all $k\in\P_{U}$, and Assumption \ref{a:muij} holds with $\mu_{mn} = \mu$ for all $(m,n)\in E(\P)$. Consider the discrete-time analog of an ISS version of \cite[Theorem 2]{Liberzon_IOSS}. A stabilizing switching signal requires to obey for all $]s:t]\in\N_{0}$
        \begin{enumerate}[label = (\roman*), leftmargin = *]
        		\item Average dwell time condition: $\displaystyle{\mathrm{N}_{\sigma}(s,t) \leq \mathrm{N}_{0} + \frac{t-s}{\tau_{a}}}$
		with $\mathrm{N}_{0}\geq 0$,\\$\tau_{a}\in]\frac{\ln\mu}{\abs{\ln\lambda_{S}}(1-\overline{\rho})-\abs{\ln\lambda_{U}}\overline{\rho}},+\infty[$, and
		\item constrained activation of non-ISS subsystems $\displaystyle{\mathrm{T}^{\mathrm{U}}(s,t) \leq \mathrm{T}_{0} + \overline{\rho}(t-s)}$
		with $\mathrm{T}_{0} \geq 0$, $\displaystyle{\overline{\rho}<\frac{\abs{\ln\lambda_{S}}}{\abs{\ln\lambda_{S}}+\abs{\ln\lambda_{U}}}}$.
        \end{enumerate}
        However, given a family of systems \eqref{e:family}, algorithmic construction of the stabilizing switching signals involves verifying both conditions (i) and (ii) for every interval of time. On the other hand, the graph-theoretic condition involved in our result is numerically easier to verify. Indeed, given the underlying weighted digraph $G(\P,E(\P))$ of the switched system \eqref{e:swsys}, algorithmic construction of the class of switching signals proposed in Theorem \ref{t:mainthm} is identical to finding a closed contractive walk $W'$ on $G(\P,E(P))$ and generating an infinite walk by repeating $W'$. Consequently, our results are more useful for constructing periodic switching signals which preserve stability of a switched system than for certifying stability of a switched system when some conditions on the rate of switches is given a priori. 
    \end{rem}
    \begin{ex}
    \rm
    	Consider a family of systems $\P = \{1,2,3,4\}$ with $\P_{S} = \{1,2\}$ and $\P_{U} = \{3,4\}$. Let all switches be admissible. Let it also be allowed to dwell on every system for two or more consecutive time steps. To construct a switching signal obeying average dwell time, we need to perform the following:\\
	1. Fix $\displaystyle{\overline{\rho}<\frac{\abs{\ln\lambda_{S}}}{\abs{\ln\lambda_{S}}+\abs{\ln\lambda_{U}}}}$, $\tau_{a}\in]\frac{\ln\mu}{\abs{\ln\lambda_{S}}(1-\overline{\rho})-\abs{\ln\lambda_{U}}\overline{\rho}},+\infty[$, $N_{0}, T_{0} \geq 0$, where $\lambda_{S} = \min\{\lambda_{1},\lambda_{2}\}$ and $\lambda_{U} = \max\{\lambda_{3},\lambda_{4}\}$.\\
	2. Verify $\displaystyle{\mathrm{N}_{\sigma}(s,t) \leq \mathrm{N}_{0} + \frac{t-s}{\tau_{a}}}$ and $\displaystyle{\mathrm{T}^{\mathrm{U}}(s,t) \leq \mathrm{T}_{0} + \overline{\rho}(t-s)}$ for {every interval} $]s:t]\in\N_{0}$ {of time}.\\
	In contrast, applying our stabilizing conditions involves two steps:\\
	1. Algorithmically detect a closed contractive walk on $G(\P,E(\P))$.\\
	2. Construct an infinite walk $\grave{a}$ la Mechanism 1.\\
	It is clear that constructing a switching signal obeying average dwell time involves checking infinitely many point-wise conditions simultaneously, whereas our conditions are finitary.
    \end{ex}
    
    \begin{rem}
    \rm
    	On the one hand, the choice of the Lyapunov-like functions $V_{i}$, $i\in\P$ and consequently, the scalars $\lambda_{i}$, $i\in\P$ and $\mu_{ij}$, $(i,j)\in E(\P)$ are not unique. On the other hand, the existence of a closed contractive walk $W$ on $G(\P,E(\P))$ depends on the choice of the above mentioned scalars. Ideally, one would like to select the Lyapunov-like functions (and consequently the scalars $\lambda_{i}$, $i\in\P$ and $\mu_{ij}$, $(i,j)\in E(\P)$) such that there exists a closed contractive walk on $G(\P,E(\P))$. However, to the best of our knowledge, the above ``co-design'' problem is difficult and in the absence of numerical solution to it, we consider the scalars under consideration (and consequently the vertex and edge weights of $G(\P,E(\P))$) to be given.
    \end{rem}
    
    In the remainder of this section we discuss algorithmic construction of a closed contractive walk on $G(\P,E(\P))$.
    \subsection{Algorithmic construction of a closed contractive walk on $G(\P,E(\P))$}
        Even though a closed contractive walk is of finite length, an upper bound on its length is not known apriori. Consequently, under what condition an algorithm that attempts to detect/design a closed contractive walk on $G(\P,E(\P))$, should stop, cannot be specified. An immediate and natural alternative is to specialize a closed contractive walk to a contractive circuit or a contractive cycle. We follow the convention: A \emph{circuit} is a closed walk in which all edges are distinct, and a \emph{cycle} is a closed walk in which all vertices are distinct except that the initial vertex = final vertex. Consequently, the length of a circuit and a cycle are at most $\abs{E(\P)}$ and $\abs{\P}$, respectively. We showed in \cite{kbnc_arxiv14} that on a given weighted digraph $G(\P,E(\P))$, the existence of a closed contractive walk, a contractive circuit, and a contractive cycle are equivalent. As a result, algorithmic construction of a contractive circuit/cycle on $G(\P,E(\P))$ serves our purpose.

        Given the underlying weighted digraph $G(\P,E(\P))$ of the switched system \eqref{e:swsys}, in \cite[Theorem 2(b) and (c)]{knc_hscc14} we proposed an algorithm for construction of a contractive circuit on $G(\P,E(\P))$. This algorithm works in two steps: The first step involves a feasibility problem (linear program) for detection of a contractive circuit on $G(\P,E(\P))$; in the second step, a contractive circuit is designed using Hierholzer's algorithm, if one such circuit exists.\footnote{The feasibility problem is based on existing shortest path problem on digraphs \cite[\S 3.4]{papa_optimization}. Given an Eulerian graph $G(V,E)$, Hierholzer's algorithm returns an Eulerian circuit \cite[p.\ 57]{Harris}.}

        In \cite{kbnc_arxiv14} we showed that the algorithmic construction of a contractive cycle on $G(\P,E(\P))$ is equivalent to finding a negative cycle on $G(\P,E(\P))$, i.e., the cycle for which the sum of the weights is less than zero. Various algorithms are available in the literature to achieve the above, e.g., the Bellman-Ford-Moore algorithm (algorithmic detection) \cite[p.\ 646]{Cormen_algo}, the negative cycle algorithm proposed in \cite{allnegcycles} (for detection and design), etc.

\section{Numerical Example}
\label{s:numex}
    In this section we present a numerical example. We consider a family of systems \eqref{e:family} with
    \begin{align*}
    	f_{1}(x,v) &= \pmat{1.05x_{2}+0.05x_{2}\exp(-\abs{x_{2}})+\exp(-\abs{x_{1}})v\\0.7x_{1}+\exp(-\abs{x_{2}})v},
	\intertext{and}
	f_{2}(x,v) &= \pmat{2x_{1}\sin(x_{1})+\exp(-\abs{x_{1}})v\\\sqrt{6}x_{2}+\exp(-\abs{x_{2}})v}.
    \end{align*}
    Clearly, $\P_{S} = \{1\}$ and $\P_{U} = \{2\}$. With the following choice of ISS-Lyapunov-like functions $V_{1}(x) = V_{2}(x) = 2x_{1}^{2}+3x_{2}^{2}$, we obtain $\lambda_{1} = 0.815, \lambda_{2} = 1.2,
	\mu_{12} = 1,  \mu_{21} = 1$.
    
    Let the following switches be allowed: $1\rightarrow 2$ and $2\rightarrow 1$. Let it also be admissible to dwell on system 2 for two (or more) consecutive time steps. We have $G(\P,E(\P))$ with $\P = \{1,2\}, E(\P) = \{(1,2),(2,1),(2,2)\}$.
    We now seek for a closed contractive walk $W'$ on the above weighted digraph $G(\P,E(\P))$. Towards this end, we apply our algorithm proposed in \cite[Theorem 2(b) and (c)]{knc_hscc14} for detection of a contractive circuit. The node (arc) incidence matrix $A$ for the above digraph $G(\P,E(\P))$ is:
    \begin{align*}
    	A = \bordermatrix{~ & (1,2) & (2,1) & (2,2') & (2',2) \cr
              1 & +1 & -1 & 0 & 0 \cr
              2 & -1 & +1 & +1 &-1 \cr
              2' &0 & 0 & -1 &+1\cr }.
    \end{align*}
    The vertex $2'$ is introduced to accommodate the self-loop at vertex $2$, see \cite[\S 3]{knc_hscc14} for a discussion on how to include self-loops in an incidence matrix. Solving the feasibility problem in \cite[Theorem 2(b)]{knc_hscc14} in the context of this example with the aid of MATLAB by employing the program YALMIP \cite{Loefberg04} and the solver SDPT3, we obtain the following solution: $\eta = (1,1 ,0 ,0)^\top$,
    with $\Xi(W') = 0.89126< 1$. Following is a circuit obtained from the vector $\eta$ with the aid of Hierholzer's algorithm: $W' = 1,2,1$.
    
    We now consider an infinite walk $W$ obtained by repeating ($\grave{a}$ la Mechanism \ref{mech:repeatmech}) the above contractive circuit $W'$. We apply the switching signal $\sigma$ corresponding ($\grave{a}$ la Fact \ref{fact:walk}) to the above infinite walk $W$ to the switched system \eqref{e:swsys}, and study the nature of $(x(t))_{t\in\N_{0}}$ for fifty different initial conditions $x(0)$ chosen uniformly at random from $[-500,500]^{2}$, and inputs $v$ chosen uniformly at random from $]0,10[$ in Figure \ref{fig:xplot}.
    
  \begin{figure}
  	\includegraphics[height = 6 cm, width = 9 cm]{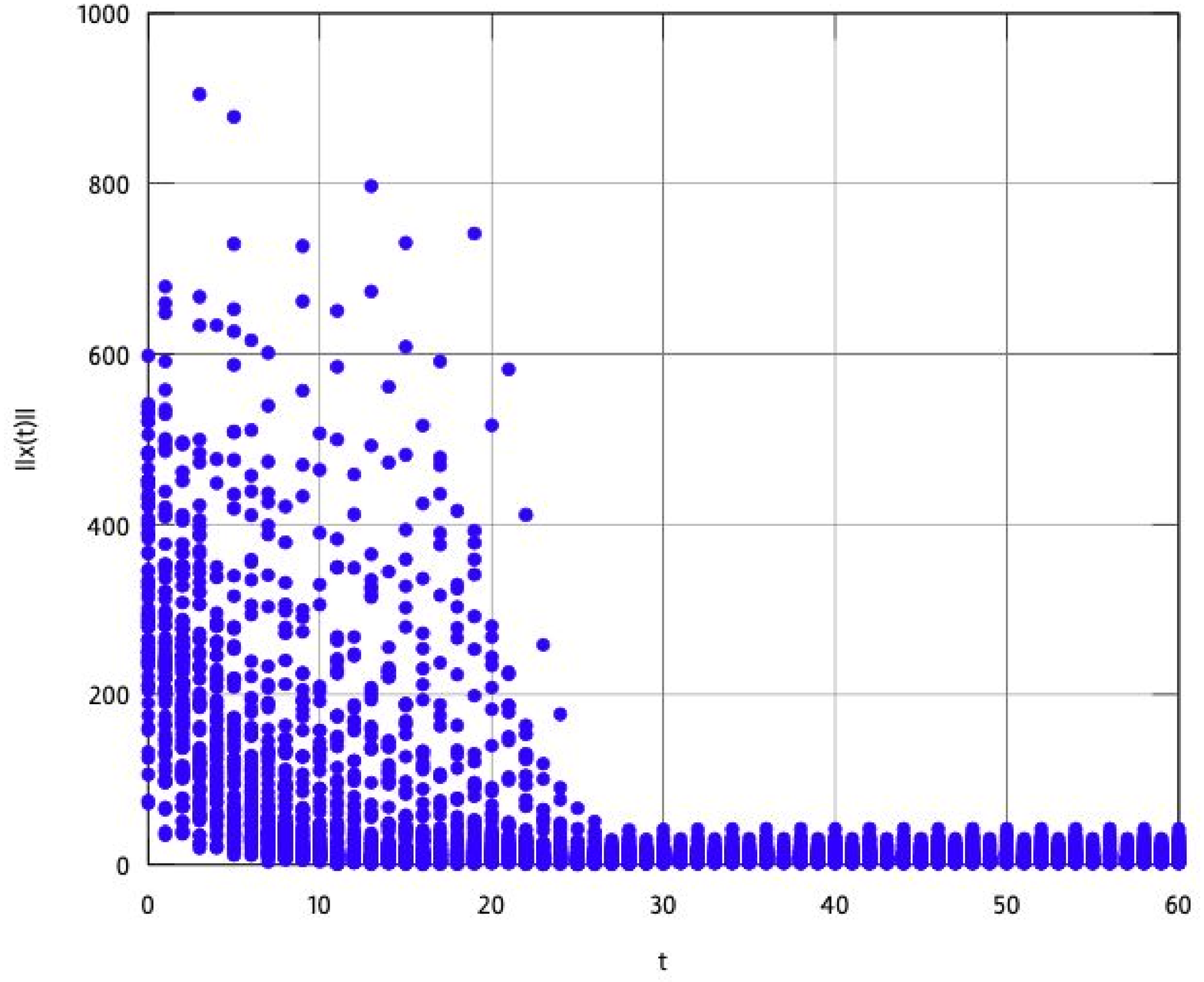}
	\caption{Plot for $(\norm{x(t)})_{t\in\N_{0}}$} \label{fig:xplot}
	

  \end{figure}
\section{Conclusion}
\label{s:concln}
	Given a family of discrete-time systems with exogenous inputs such that not all subsystems are ISS, in this article we presented a class of switching signals under which the resulting switched system is ISS. We employed graph-theoretic arguments in our analysis. A weighted digraph is associated to the given family of systems and the admissible transitions, and the switching signal is expressed as an infinite walk on the weighted digraph. Our stabilizing switching signals are periodic in nature in the sense that they correspond to infinite walks constructed by repeating suitable closed walks on the above weighted digraph. Consequently, these switching signals afford simple algorithmic construction. 
	
	On the one hand, the proposed stability condition requires presence of at least one ISS system in the family and hence does not cater to the families in which all systems are unstable. On the other hand, we do not require the unstable systems to form a stable combination. We conjecture that the class of switching signals discussed in this article readily extends to input/output-to-state stability (IOSS) of discrete-time switched systems.
	
	Moreover, we observe that our stabilizing switching signals are ``not necessarily'' periodic. If the underlying weighted digraph of a switched system admits multiple closed contractive walks which can be concatenated, then it is possible to generate aperiodic infinite walks that correspond to stabilizing switching signals. A detailed analysis for this will be reported elsewhere.
    
\section{Proof of Theorem \ref{t:mainthm}}
\label{s:proofs}
	We first catalog the following lemma, which will be utilized in our proof of Theorem \ref{t:mainthm}.
	
	\begin{lem}
	\label{lem:contrawalk}
		Consider the underlying weighted digraph $G(\P,E(\P))$ of the switched system \ref{e:swsys}, and an infinite walk $W$ that is constructed by repeating ($\grave{a}$ la Mechanism \ref{mech:repeatmech}) a closed contractive walk $W'$ on $G(\P,E(\P))$. Every closed sub walk $W''$ of the above infinite walk $W$ such that $\abs{W'} = \abs{W''}$, satisfies the following:\\
		i) $W''$ is contractive, and\\
		ii) $\displaystyle{\Xi(W'') = \Xi(W')}$.
	\end{lem}
	The above lemma follows from the observation that the closed sub walk $W''$ of $W$ is nothing but a rotated version of the closed contractive walk $W'$. Consequently, i) and ii) follow at once.
	
	We are now in a position to present our 
	\begin{proof}[Proof of Theorem \ref{t:mainthm}]
		Fix $t\in\N$. Let $N_{\sigma}(0,t)$ be the number of switches of $\sigma$ before (and including) $t$, and let $0 =: \tau_{0}<\tau_{1}<\cdots$ be the corresponding switching instants.
		
		Applying \eqref{e:dLyapprop} and \eqref{e:muijineq}, we obtain
		\begin{align}
		\label{e:proof1}
			V_{\sigma(t)}(x(t)) \leq \psi_{1}(t)V_{\sigma(0)}(x(0)) + \gamma(\norm{v}_{t})\psi_{2}(t),
		\end{align}
		where
		\begin{align}
		\label{e:psi1defn}
			\psi_{1}(t) := \Biggl(\prod_{\substack{i=0\\\tau_{\Ntsigma+1}:=t}}^{\Ntsigma}\lambda_{\sigma(\tau_{i})}^{\tau_{i+1}-\tau_{i}}\cdot\prod_{i=0}^{\Ntsigma-1}\mu_{\sigma(\tau_{i})\sigma(\tau_{i+1})}\Biggr),
		\end{align}
		and
		\begin{align}
		\label{e:psi2defn}
			\psi_{2}(t) := \Biggl(\sum_{\substack{i=0\\\tau_{\Ntsigma+1}:=t}}^{\Ntsigma}\Biggl(\prod_{j=i+1}^{\Ntsigma}\lambda_{\sigma(\tau_{j})}^{\tau_{j+1}-\tau_{j}}\cdot\prod_{j=i+1}^{\Ntsigma-1}\mu_{\sigma(\tau_{j})\sigma(\tau_{j+1})}\cdot\sum_{k=0}^{(\tau_{i+1}-\tau_{i})-1}\lambda_{\sigma(\tau_{i})}^{k}\Biggr)\Biggr).
		\end{align}
		In view of \eqref{e:Lyapprop}, we obtain
		\begin{align}
		\label{e:proof2}
			\underline{\alpha}(\norm{x(t)}) \leq \psi_{1}(t)\overline{\alpha}(\norm{x(0)}) + \gamma(\norm{v}_{t})\psi_{2}(t).
		\end{align}
		By Definition \ref{d:iss}, for ISS of \eqref{e:swsys}, we need to show the following:\\
		i) $\psi_{1}(\cdot)$ is bounded above by a class $\mathcal{L}$ function, and\\
		ii) $\psi_{2}(\cdot)$ is bounded.\\\\
		i) In the absence of inputs, i.e., $v \equiv 0$, the switching signal under consideration guarantees GAS of \eqref{e:swsys} \cite[Theorem 1]{knc_hscc14}. Consequently, i) is verified.\\
		ii) Observe that $\psi_{2}(t)$ can be rewritten as 
		\begin{align*}
			\psi_{2}(t) &= \sum_{k\in\P}\sum_{\substack{i:\sigma(\tau_{i}) = k\\i = 0,\cdots, \Ntsigma\\\tau_{\mathrm{N}_{\sigma(0,t)+1}:=t}}}\Biggl(\prod_{j=i+1}^{\Ntsigma}\lambda_{\sigma(\tau_{j})}^{\tau_{j+1}-\tau_{j}}\cdot\prod_{j=i+1}^{\Ntsigma-1}\mu_{\sigma(\tau_{j})\sigma(\tau_{j+1})}\cdot\frac{1-\lambda_{k}^{\tau_{i+1}-\tau_{i}}}{1-\lambda_{k}}\Biggr)\nonumber\\
			&=\sum_{k\in\P_{S}}\frac{1}{1-\lambda_{k}}\sum_{\substack{i:\sigma(\tau_{i}) = k\\i = 0,\cdots, \Ntsigma\\\tau_{\mathrm{N}_{\sigma(0,t)+1}:=t}}}\Biggl(\prod_{j=i+1}^{\Ntsigma}\lambda_{\sigma(\tau_{j})}^{\tau_{j+1}-\tau_{j}}\cdot\prod_{j=i+1}^{\Ntsigma-1}\mu_{\sigma(\tau_{j})\sigma(\tau_{j+1})}\cdot{\bigl(1-\lambda_{k}^{\tau_{i+1}-\tau_{i}}\bigr)}\Biggr)\nonumber\\
			&\:\:+\sum_{\ell\in\P_{U}}\frac{1}{1-\lambda_{\ell}}\sum_{\substack{i:\sigma(\tau_{i}) = \ell\\i = 0,\cdots, \Ntsigma\\\tau_{\mathrm{N}_{\sigma(0,t)+1}:=t}}}\Biggl(\prod_{j=i+1}^{\Ntsigma}\lambda_{\sigma(\tau_{j})}^{\tau_{j+1}-\tau_{j}}\cdot\prod_{j=i+1}^{\Ntsigma-1}\mu_{\sigma(\tau_{j})\sigma(\tau_{j+1})}\cdot{\bigl(1-\lambda_{\ell}^{\tau_{i+1}-\tau_{i}}\bigr)}\Biggr). 
		\end{align*}
		Since $\displaystyle{0 < \lambda_{k} < 1}$, $k\in\P_{S}$ and $\lambda_{\ell} > 1$, $\ell\in\P_{U}$, we have the right-hand side of the above quantity is equal to
		\begin{align}
		\label{e:proof4}
			&\sum_{k\in\P_{S}}\frac{1}{\abs{1-\lambda_{k}}}\sum_{\substack{i:\sigma(\tau_{i}) = k\\i = 0,\cdots, \Ntsigma\\\tau_{\mathrm{N}_{\sigma(0,t)+1}:=t}}}\Biggl(\prod_{j=i+1}^{\Ntsigma}\lambda_{\sigma(\tau_{j})}^{\tau_{j+1}-\tau_{j}}\cdot\prod_{j=i+1}^{\Ntsigma-1}\mu_{\sigma(\tau_{j})\sigma(\tau_{j+1})}\cdot{\bigl(1-\lambda_{k}^{\tau_{i+1}-\tau_{i}}\bigr)}\Biggr)\nonumber\\
			&-\sum_{\ell\in\P_{U}}\frac{1}{\abs{1-\lambda_{\ell}}}\sum_{\substack{i:\sigma(\tau_{i}) = \ell\\i = 0,\cdots, \Ntsigma\\\tau_{\mathrm{N}_{\sigma(0,t)+1}:=t}}}\Biggl(\prod_{j=i+1}^{\Ntsigma}\lambda_{\sigma(\tau_{j})}^{\tau_{j+1}-\tau_{j}}\cdot\prod_{j=i+1}^{\Ntsigma-1}\mu_{\sigma(\tau_{j})\sigma(\tau_{j+1})}\cdot{\bigl(1-\lambda_{\ell}^{\tau_{i+1}-\tau_{i}}\bigr)}\Biggr)\nonumber\\
			&\leq \sum_{k\in\P_{S}}\frac{1}{\abs{1-\lambda_{k}}}\sum_{\substack{i:\sigma(\tau_{i}) = k\\i = 0,\cdots, \Ntsigma\\\tau_{\mathrm{N}_{\sigma(0,t)+1}:=t}}}\Biggl(\prod_{j=i+1}^{\Ntsigma}\lambda_{\sigma(\tau_{j})}^{\tau_{j+1}-\tau_{j}}\cdot\prod_{j=i+1}^{\Ntsigma-1}\mu_{\sigma(\tau_{j})\sigma(\tau_{j+1})}\Biggr)\nonumber\\
			&-\sum_{\ell\in\P_{U}}\frac{1}{\abs{1-\lambda_{\ell}}}\sum_{\substack{i:\sigma(\tau_{i}) = \ell\\i = 0,\cdots, \Ntsigma\\\tau_{\mathrm{N}_{\sigma(0,t)+1}:=t}}}\Biggl(\prod_{j=i+1}^{\Ntsigma}\lambda_{\sigma(\tau_{j})}^{\tau_{j+1}-\tau_{j}}\cdot\prod_{j=i+1}^{\Ntsigma-1}\mu_{\sigma(\tau_{j})\sigma(\tau_{j+1})}\cdot\lambda_{\ell}^{\tau_{i+1}-\tau_{i}}\Biggr).
		\end{align}
		Also,
		\begin{align}
		\label{e:proof5}
			\ln\Biggl(\prod_{j=i+1}^{\Ntsigma-1}\mu_{\sigma(\tau_{j})\sigma(\tau_{j+1})}\Biggr) &= \sum_{j=i+1}^{\Ntsigma-1}\ln\mu_{\sigma(\tau_{j})\sigma(\tau_{j+1})} \nonumber\\
			&=\sum_{k\in\P_{S}}\sum_{j=i+1}^{\Ntsigma-1}\sum_{\substack{k\rightarrow\ell:\\(k,\ell)\in E(\P)}}\ln\mu_{k\ell} \nonumber\\
			&=\sum_{(k,\ell)\in E(\P)}(\ln\mu_{k\ell})\#\{k\rightarrow\ell\}_{\tau_{i+1}}^{t},
		\end{align}
		where $\#\{k\rightarrow\ell\}_{s}^{t}$ denotes the number of times a switch from system $k$ to system $\ell$ has occurred in the interval $]s:t]\subset\N_{0}$, and
		\begin{align}
		\label{e:proof6}
		\ln\Biggl(\prod_{j=i+1}^{\Ntsigma}\lambda_{\sigma(\tau_{j})}^{\tau_{j+1}-\tau_{j}}\Biggr) &= \sum_{j=i+1}^{\Ntsigma}(\tau_{j+1}-\tau_{j})\ln\lambda_{\sigma(\tau_{j})} \nonumber\\
		&= \sum_{j=i+1}^{\Ntsigma}\Bigl(\sum_{k\in\P}1_{\{\sigma(\tau_{j})=k\}}(\tau_{j+1}-\tau_{j})\ln\lambda_{k}\Bigr) \nonumber\\
		&\hspace*{-1.5cm}= \sum_{k\in\P_{S}}\ln\lambda_{k}\sum_{\substack{j:\sigma(\tau_{j})=k\\j = i+1,\cdots,\Ntsigma}}(\tau_{j+1}-\tau_{j}) 
		+\sum_{\ell\in\P_{U}}\ln\lambda_{\ell}\sum_{\substack{j:\sigma(\tau_{j})=\ell\\j = i+1,\cdots,\Ntsigma}}(\tau_{j+1}-\tau_{j}) \nonumber\\
		&\hspace*{-1.5cm}= -\sum_{k\in\P_{S}}\abs{\ln\lambda_{k}}\sum_{\substack{j:\sigma(\tau_{j})=k\\j = i+1,\cdots,\Ntsigma}}(\tau_{j+1}-\tau_{j}) 
		+\sum_{\ell\in\P_{U}}\abs{\ln\lambda_{\ell}}\sum_{\substack{j:\sigma(\tau_{j})=\ell\\j = i+1,\cdots,\Ntsigma}}(\tau_{j+1}-\tau_{j}) \nonumber\\
		&\hspace*{-1.5cm}=-\sum_{k\in\P_{S}}\abs{\ln\lambda_{k}}\#\{k\}_{\tau_{i+1}}^{t} + \sum_{\ell\in\P_{U}}\abs{\ln\lambda_{\ell}}\#\{\ell\}_{\tau_{i+1}}^{t},
		\end{align}
		where $\#\{k\}_{s}^{t}$ denotes the number of times a system $k$ is activated in the interval $]s:t]\subset\N_{0}$.\\
		In view of \eqref{e:proof5} and \eqref{e:proof6}, it is immediate that the quantity in \eqref{e:proof4} is at most equal to
		\begin{align}
		\label{e:proof7}
			&\sum_{k\in\P_{S}}\frac{1}{\abs{1-\lambda_{k}}}\sum_{\substack{i:\sigma(\tau_{i})=k\\i=0,\cdots,\Ntsigma\\\tau_{\Ntsigma+1}:=t}}\exp\Biggl(-\sum_{p\in\P_{S}}\abs{\ln\lambda_{p}}\#\{p\}_{\tau_{i+1}}^{t}\nonumber\\
			&+\sum_{q\in\P_{U}}\abs{\ln\lambda_{q}}\#\{q\}_{\tau_{i+1}}^{t} + \sum_{(m,n)\in E(\P)}(\ln\mu_{mn})\#\{m\rightarrow n\}_{\tau_{i+1}}^{t}\Biggr) \nonumber\\
			&\sum_{\ell\in\P_{U}}\frac{1}{\abs{1-\lambda_{\ell}}}\sum_{\substack{i:\sigma(\tau_{i})=\ell\\i=0,\cdots,\Ntsigma\\\tau_{\Ntsigma+1}:=t}}\exp\Biggl(-\sum_{p\in\P_{S}}\abs{\ln\lambda_{p}}\#\{p\}_{\tau_{i}}^{t}\nonumber\\
			&+\sum_{q\in\P_{U}}\abs{\ln\lambda_{q}}\#\{q\}_{\tau_{i}}^{t} + \sum_{(m,n)\in E(\P)}(\ln\mu_{mn})\#\{m\rightarrow n\}_{\tau_{i}}^{t}\Biggr).
		\end{align}
		Let for an interval $]s:t]\subset\N_{0}$ of time,
		\begin{align}
		\label{e:gdefn}
			g(s,t) &:= -\sum_{k\in\P_{S}}\abs{\ln\lambda_{k}}\#\{k\}_{s}^{t} + \sum_{\ell\in\P_{U}}\abs{\ln\lambda_{\ell}}\#\{k\}_{s}^{t}
			+ \sum_{(m,n)\in E(\P)}(\ln\mu_{mn})\#\{m\rightarrow n\}_{s}^{t}.
		\end{align} 
		Applying \eqref{e:gdefn}, \eqref{e:proof7} can be rewritten as
		\begin{align}
			\psi_{2}(t) &= \sum_{k\in\P_{S}}\frac{1}{\abs{1-\lambda_{k}}}\sum_{\substack{i:\sigma(\tau_{i}) = k\\i=0,\cdots,\Ntsigma\\\tau_{\Ntsigma+1}:=t}}\exp\bigl(g(\tau_{i+1},t)\bigr) 
			+ \sum_{\ell\in\P_{U}}\frac{1}{\abs{1-\lambda_{\ell}}}\sum_{\substack{i:\sigma(\tau_{i}) = \ell\\i=0,\cdots,\Ntsigma\\\tau_{\Ntsigma+1}:=t}}\exp\bigl(g(\tau_{i},t)\bigr) \nonumber\\
			&\leq \sum_{k\in\P_{S}}\frac{1}{\abs{1-\lambda_{k}}}\sum_{i=0}^{\Ntsigma}\exp\bigl(g(\tau_{i+1},t)\bigr) 
			+\sum_{\ell\in\P_{U}}\frac{1}{\abs{1-\lambda_{\ell}}}\sum_{i=0}^{\Ntsigma}\exp\bigl(g(\tau_{i},t)\bigr).
		\end{align}
		We now concentrate on the quantity $\displaystyle{\sum_{i=0}^{\Ntsigma}\exp\bigl(g(\tau_{i},t)\bigr)}$. Recall that our switching signal corresponds to an infinite walk $W$ constructed by repeating a closed contractive walk $W'$ on $G(\P,E(\P))$. Let $t > n\abs{W'}$ for some $n \geq 0$ and $\Xi(W') = -\varepsilon$ for some $\varepsilon > 0$. By construction of $W$, the following are immediate:
         \begin{itemize}[label = $\circ$, leftmargin = *]
            \item For $0 \leq \tau_{i} \leq t-n\abs{W'}$, there are $n$ closed contractive walks of length $\abs{W'}$ between $\tau_{i}$ and $t$.
            \item For $t-n\abs{W'}+1 \leq \tau_{i} \leq \abs{W'}$, there are $(n-1)$ closed contractive walks of length $\abs{W'}$ between $\tau_{i}$ and $t$.
            \item For $\abs{W'}+1 \leq \tau_{i} \leq \abs{W'}+(t-n\abs{W'})$, there are $(n-1)$ closed contractive walks of length $\abs{W'}$ between $\tau_{i}$ and $t$.
            \item For $\abs{W'}+(t-n\abs{W'})+1 \leq \tau_{i} \leq 2\abs{W'}$, there are $(n-2)$ closed contractive walks of length $\abs{W'}$ between $\tau_{i}$ and $t$.
            \item For $2\abs{W'}+1 \leq \tau_{i} \leq 2\abs{W'}+(t-n\abs{W'})$, there are $(n-2)$ closed contractive walks of length $\abs{W'}$ between $\tau_{i}$ and $t$.
            \item For $2\abs{W'}+(t-n\abs{W'})+1 \leq \tau_{i} \leq 3\abs{W'}$, there are $(n-3)$ closed contractive walks of length $\abs{W'}$ between $\tau_{i}$ and $t$.
            \item For $3\abs{W'}+1 \leq \tau_{i} \leq 3\abs{W'}+(t-n\abs{W'})$, there are $(n-3)$ closed contractive walks of length $\abs{W'}$ between $\tau_{i}$ and $t$.
            \item $\vdots$
         \end{itemize}
         We have
         \begin{align}
         \label{e:proof9}
            \sum_{i=0}^{\Ntsigma}\exp\bigl(g(\tau_{i},t)\bigr) &= \sum_{\tau_{i}:0\leq\tau_{i}\leq n\abs{W'}}\exp\bigl(g(\tau_{i},t)\bigr) + \sum_{\tau_{i}:(n\abs{W'}+1)\leq\tau_{i}\leq t}\exp\bigl(g(\tau_{i},t)\bigr).
         \end{align}
         Let 
         \begin{align*}
         	\displaystyle{a := \max_{i,(i,j),j\subset W'}\abs{(\ln\mu_{ij})+\abs{\ln\lambda_{j}}}}.
	\end{align*}
         \begin{align}
         \label{e:proof11}
            \sum_{\tau_{i}:0\leq\tau_{i}\leq n\abs{W'}}\exp\bigl(g(\tau_{i},t)\bigr) &= \sum_{\tau_{i}:0\leq\tau_{i}\leq (t-n\abs{W'})}\exp\bigl(g(\tau_{i},t)\bigr)\nonumber\\
            &\hspace*{-3cm}+\sum_{\tau_{i}:(t-n\abs{W'})+1\leq\tau_{i}\leq \abs{W'}}\exp\bigl(g(\tau_{i},t)\bigr)
            + \sum_{\tau_{i}:(\abs{W'}+1)\leq\tau_{i}\leq \abs{W'}+(t-n\abs{W'})}\exp\bigl(g(\tau_{i},t)\bigr)\nonumber\\
            &\hspace*{-3cm}+\sum_{\tau_{i}:\abs{W'}+(t-n\abs{W'})+1\leq\tau_{i}\leq 2\abs{W'}}\exp\bigl(g(\tau_{i},t)\bigr)+\sum_{\tau_{i}:(2\abs{W'}+1)\leq\tau_{i}\leq 2\abs{W'}+(t-n\abs{W'})}\exp\bigl(g(\tau_{i},t)\bigr)\nonumber\\
            &\hspace*{-3cm}+\sum_{\tau_{i}:2\abs{W'}+(t-n\abs{W'})+1\leq\tau_{i}\leq 3\abs{W'}}\exp\bigl(g(\tau_{i},t)\bigr)+\cdots\nonumber\\
            &\hspace*{-3cm}+\sum_{\tau_{i}:((n-1)\abs{W'}+1)\leq\tau_{i}\leq (n-1)\abs{W'}+(t-n\abs{W'})}\exp\bigl(g(\tau_{i},t)\bigr)\nonumber\\
            &\hspace*{-3cm}+\sum_{\tau_{i}:(n-1)\abs{W'}+(t-n\abs{W'})+1\leq\tau_{i}\leq n\abs{W'}}\exp\bigl(g(\tau_{i},t)\bigr).
         \end{align}
         Now,
         \begin{align*}
            \sum_{\tau_{i}:0\leq\tau_{i}\leq (t-n\abs{W'})}\exp\bigl(g(\tau_{i},t)\bigr)&= \sum_{\tau_{i}:0\leq\tau_{i}\leq (t-n\abs{W'})}\exp\bigl(-n\varepsilon+g(\tau_{i}+n\abs{W'},t)\bigr)\nonumber\\
            &\leq \sum_{\tau_{i}:0\leq\tau_{i}\leq (t-n\abs{W'})}\exp\bigl(-n\varepsilon+(\abs{W'}-1)a\bigr)\nonumber\\
            &\leq (t-n\abs{W'})\exp\bigl(-n\varepsilon+(\abs{W'}-1)a\bigr)\nonumber\\
            &\leq (\abs{W'}-1)\exp\bigl(-n\varepsilon+(\abs{W'}-1)a\bigr),
         \end{align*}
         \begin{align*}
            &\sum_{\tau_{i}:(t-n\abs{W'})+1\leq\tau_{i}\leq \abs{W'}}\exp\bigl(g(\tau_{i},t)\bigr)\nonumber\\&= \sum_{\tau_{i}:(t-n\abs{W'})+1\leq\tau_{i}\leq \abs{W'}}\exp\bigl(-(n-1)\varepsilon+g(\tau_{i}+(n-1)\abs{W'},t)\bigr)\nonumber\\
            &\leq \sum_{\tau_{i}:(t-n\abs{W'})+1\leq\tau_{i}\leq \abs{W'}}\exp\bigl(-(n-1)\varepsilon+(\abs{W'}-1)a\bigr)\nonumber\\
            &\leq (\abs{W'}-(t-n\abs{W'})-1)\exp\bigl(-(n-1)\varepsilon+(\abs{W'}-1)a\bigr)\nonumber\\
            &\leq (\abs{W'}-1)\exp\bigl(-(n-1)\varepsilon+(\abs{W'}-1)a\bigr).
         \end{align*}
         Similarly,
         \begin{align*}
            \sum_{\tau_{i}:(\abs{W'}+1)\leq\tau_{i}\leq \abs{W'}+(t-n\abs{W'})}\exp\bigl(g(\tau_{i},t)\bigr)&\leq (\abs{W'}-2)\exp\bigl(-(n-1)\varepsilon+(\abs{W'}-2)a\bigr),
         \end{align*}
         \begin{align*}
            \sum_{\tau_{i}:(\abs{W'}+(t-n\abs{W'})+1)\leq\tau_{i}\leq 2\abs{W'}}\exp\bigl(g(\tau_{i},t)\bigr)&\leq (\abs{W'}-1)\exp\bigl(-(n-2)\varepsilon+(\abs{W'}-1)a\bigr),
         \end{align*}
         \begin{align*}
            \sum_{\substack{\tau_{i}:(2\abs{W'}+1)\leq\tau_{i}\leq 2\abs{W'}+(t-n\abs{W'})}}\exp\bigl(g(\tau_{i},t)\bigr)&\leq (\abs{W'}-2)\exp\bigl(-(n-2)\varepsilon+(\abs{W'}-2)a\bigr),
         \end{align*}
         \begin{align*}
            \sum_{\tau_{i}:(2\abs{W}'+(t-n\abs{W'})+1)\leq\tau_{i}\leq 3\abs{W'}}\exp\bigl(g(\tau_{i},t)\bigr)&\leq (\abs{W'}-1)\exp\bigl(-(n-3)\varepsilon+(\abs{W'}-1)a\bigr),\\
            \vdots \nonumber
         \end{align*}
         \begin{align*}
            \sum_{\tau_{i}:((n-1)\abs{W'}+1)\leq\tau_{i}\leq (n-1)\abs{W'}+(t-n\abs{W'})}\exp\bigl(g(\tau_{i},t)\bigr)&\leq (\abs{W'}-2)\exp\bigl(-\varepsilon+(\abs{W'}-2)a\bigr),
         \end{align*}
         \begin{align*}
            \sum_{\tau_{i}:((n-1)\abs{W'}+(t-n\abs{W'})+1)\leq\tau_{i}\leq n\abs{W'}}\exp\bigl(g(\tau_{i},t)\bigr)&\leq (\abs{W'}-1)\exp\bigl((\abs{W'}-1)a\bigr),
         \end{align*}
         and
         \begin{align*}
            \sum_{\tau_{i}:(n\abs{W'}+1)\leq\tau_{i}\leq t}\exp\bigl(g(\tau_{i},t)\bigr)&= \sum_{\tau_{i}:(n\abs{W'}+1)\leq\tau_{i}\leq t}\exp\bigl(-0\cdot\varepsilon+g(\tau_{i}+0\cdot\abs{W'},t)\bigr)\nonumber\\
            &\leq \sum_{\tau_{i}:(n\abs{W'}+1)\leq\tau_{i}\leq t}\exp\bigl((\abs{W'}-2)a\bigr)\nonumber\\
            &\leq (t-(n\abs{W'}+1))\exp\bigl((\abs{W'}-2)a\bigr)\nonumber\\
            &\leq (\abs{W'}-2)\exp\bigl((\abs{W'}-2)a\bigr).
         \end{align*}
         Consequently,
         \begin{align}
         \label{e:proof21}
            \sum_{i=0}^{\Ntsigma}\exp\bigl(g(\tau_{i},t)\bigr)
            &\leq (\abs{W'}-1)\exp\bigl((\abs{W'}-1)a\bigr)\cdot(1+\exp(-\varepsilon)+\cdots+\exp(-n\varepsilon))\nonumber\\
            &\hspace*{-1.5cm}+ (\abs{W'}-2)\exp\bigl((\abs{W'}-2)a\bigr)\cdot(1+\exp(-\varepsilon)+\cdots+\exp(-(n-1)\varepsilon))\nonumber\\
            &\hspace*{-1.5cm}\leq (\abs{W'}-1)\exp\bigl((\abs{W'}-1)a\bigr)\cdot\frac{1-\exp(-(n+1)\varepsilon)}{1-\exp(-\varepsilon)}\nonumber\\
            &\hspace*{-1.2cm}+ (\abs{W'}-2)\exp\bigl((\abs{W'}-2)a\bigr)\cdot\frac{1-\exp(-n\varepsilon)}{1-\exp(-\varepsilon)}\nonumber\\
            &\hspace*{-1.5cm}<\frac{1}{1-\exp(-\varepsilon)}\Bigl((\abs{W'}-1)\exp\bigl((\abs{W'}-1)a\bigr)+(\abs{W'}-2)\exp\bigl((\abs{W'}-2)a\bigr)\Bigr).
         \end{align}
         Recall that the sets $\P_{S}$ and $\P_{U}$ are finite. Therefore, ii) holds. This completes our proof for Theorem \ref{t:mainthm}.
		\end{proof}
    
 


\end{document}